\documentclass[preprint,amsmath]{revtex4}
\usepackage{graphicx,epsfig,dsfont,amssymb,amsthm,amsfonts,amsbsy,mathrsfs,amscd,appendix}  
\usepackage{amsmath}
\usepackage{array}
\usepackage{tabularx}
\usepackage{booktabs}
\usepackage{enumerate}
\usepackage[colorlinks,linkcolor=blue,anchorcolor=blue,citecolor=blue]{hyperref}
\usepackage{graphicx}
\usepackage{subfig}
\newtheorem{thm}{Theorem}
\newtheorem{cor}{Corollary}
\newtheorem{lem}{Lemma}
\newtheorem{example}{Example}

\newcounter{appendixsection}
\renewcommand{\theappendixsection}{\Alph{appendixsection}}

\makeatletter
\newcommand{\appendixsection}[1]{%
  \refstepcounter{appendixsection}%
  \par
  \vskip 3.5ex plus 1ex minus .2ex
  \noindent
  {\normalsize\bfseries Appendix \theappendixsection. #1}%
  \nobreak
  \vskip 2.3ex plus .2ex
  \addcontentsline{toc}{section}{Appendix \theappendixsection. #1}%
  \hypertarget{appsec:\theappendixsection}{}%
  \label{appsec:\theappendixsection}%
}
\makeatother

\begin{document}

\begin{center}
		\bf{Imaginarity measures induced by real part states and the complementarity relations}

\end{center}

\begin{center}
    Jingyan Liu$~^{1}$, Yue Sun$~^{2}$,  Jianwei Xu$~^{3,*}$, Ming-Jing Zhao$~^{1,\dagger}$

 \small $^{1}$School of Science, Beijing Information Science and Technology University, Beijing, 102206, P. R. China \\
 \small $^{2}$School of Mathematics, Nanjing University of Aeronautics and Astronautics, Nanjing, 210016, P. R. China \\
  \small $~^{3}$ {School of Mathematics and Statistics, Shaanxi Normal University, Xi¡¯an, 710119, P. R. China\\}

 \small $^{*}$Corresponding author: {xujianwei@snnu.edu.cn}  \\
 \small $^{\dagger}$Corresponding author: {zhaomingjingde@126.com}

\end{center}

Complex numbers are indispensable in quantum mechanics and the resource theory of imaginarity has been developed recently. In this paper, we propose a method to construct imaginary measures by real part states. Specifically, we propose an imaginarity measure in terms of fidelity and explore its properties. The analytical expression of the imaginarity measure is presented in qubit systems.
The relations between the proposed imaginarity measure and some other imaginarity measures (such as geometric imaginarity, Tsallis relative entropy imaginarity and trace norm imaginarity) are derived. The complementarity relations of the imaginarity measure under a complete set of mutually unbiased bases are provided in low-dimensional systems. This work not only highlights the prominent role of the real part state in the imaginarity resource theory, but also reveals  the constraint of imaginarity on a complete set of mutually unbiased bases physically.

\section{Introduction}
Complex numbers are indispensable in quantum mechanics. The  finding about the experimental scenario that cannot be exactly modeled using quantum theory with real-valued amplitudes only \cite{Renou},
attracts much attention on the study of the imaginarity \cite{ZDLi,MCChen,DWu}. It was
shown that imaginarity has crucial
effects on certain discrimination tasks \cite{KDWu}, hiding and masking \cite{HZhu},  machine learning
\cite{Sajjan}, pseudorandomness \cite{Haug}, outcome statistics of linear optical experiments \cite{Jones}, Kirkwood-Dirac quasiprobability
distributions \cite{Budiyono,Budiyono-Dipojono,Budiyono-Agusta,Wagner}, weak-value theory \cite{Wagner-Gal}.

Along with these extensive applications, the imaginarity resource theory has been formulated theoretically \cite{A.H.2018}.
This theory identifies the resource as the presence of nonvanishing imaginary components in a quantum state's density matrix, emphasizing the fundamental role of complex numbers in quantum mechanics.
To quantify the imaginarity and
determine the resourcefulness of a given quantum state, different measures of imaginarity have been proposed, for example, the trace norm imaginarity \cite{A.H.2018}, the relative entropy imaginarity  \cite{S.N.X.2021}, the geometric  imaginarity  \cite{T.V.K2023}, the Tsallis relative entropy  imaginarity \cite{J.W.X.2024} and the imaginarity measure based on the generalized quantum Jensen-Shannon divergence \cite{P.Tian2025}. Besides to these explicit imaginarity measures, Ref. \cite{S.P.Du2025} develops two methods to construct imaginarity measures, one is by the convex roof construction and the other is by the least imaginarity of the input pure states under real operations.

In the imaginarity resource theory, for any quantum state $\rho$, it can induce the real part state  ${\rm Re}(\rho)$ which is derived by eliminating the imaginary part of $\rho$, that is, ${\rm Re}(\rho)=\frac{1}{2}(\rho+\rho^*)$ with superscript $*$ denoting the conjugation. A natural question is whether one can measure the imaginarity of an arbitrary state using its real part state  ${\rm Re}(\rho)$ straightforwardly. The answer appears to be affirmative, as further derivation of the original formulas for certain existing imaginarity measures-such as the trace norm measure \cite{A.H.2018} and the relative entropy measure \cite{S.N.X.2021}-reveals that  quantifying imaginarity by utilizing the real part state is available and reasonable. Following by this, another question is whether this kind of imaginarity measures is just valid for such special cases.  In this work, we develop a method to construct imaginarity measures by the real part state ${\rm Re}(\rho)$ generally. This indicates the real part state ${\rm Re}(\rho)$ characterizes the imaginarity essentially.
By this way, we propose a specific imaginarity measure and study its properties thoroughly. Further, we find the complementarity relation of imaginarity under a complete set of mutually unbiased bases (MUBs) in low-dimensional systems. This relation reveals the intrinsic property of imaginarity as the quantum resource physically.

The remainder of this paper is organized as follows. In section II, we give a review of fundamental concepts regarding to imaginarity measures.
In section III, we  develop a method to construct imaginarity measures  by real part state ${\rm Re}(\rho)$. In particular,
we propose the imaginarity measure $M_{\mathrm{Re}}$  in terms of fidelity and study its properties.
We compare the imaginarity measure $M_{\mathrm{Re}}$ with geometric imaginarity measure, Tsallis relative entropy imaginarity and trace norm imaginarity in turn.
In section IV, we show  a complementarity relation for the imaginarity measure $M_{\mathrm{Re}}$ under a complete set of MUBs in low-dimensional systems. The conclusion is given in section V.

\section{Preliminaries}

Let  $\left\{|j\rangle \right \}_{j=0}^{d-1}$ be an orthonormal basis for $d$-dimensional Hilbert space $\mathcal{H}$.
In imaginarity resource theory, the quantum state $\rho$ is real state, if  the density matrix of $\rho$ under the reference basis is real, that is, $\rho=\displaystyle\sum_{ij}\rho_{ij}|i\rangle\langle j|$, where each $\rho_{ij}=\langle i|\rho |j \rangle \in \mathbb{R}$. The set of real states is denoted as $\cal{R}$. Let $\Phi$ be a quantum operation with Kraus operators ${K_{l}}$, $\Phi(\cdot)=\sum_l K_{l} (\cdot) K_{l}^\dagger$.
$\Phi$ is called a real operation if $\langle m|K_{l}|n\rangle\in \mathbb{R}$ for all $l$, $m$ and $n$. A nonnegative function $M$ on quantum states is called an imaginarity measure if $M$ satisfies the following five conditions \cite{A.H.2018, WuKD2021, WuKD20212},

(M1) Faithfulness: $M(\rho)\geq0$ and $M(\rho)=0$ if and only if $\rho$ is real.

(M2) Monotonicity: $M(\Phi(\rho))\leq M(\rho)$ if $\Phi$ is a real operation.

(M3) Probabilistic monotonicity: $\sum_{l}\mathrm{Tr}(K_{l}\rho K_{l}^{T})M\left[\frac{K_{l}\rho K_{l}^{T}}{\mathrm{Tr}(K_{l}\rho K_{l}^{T})}\right]\leq M(\rho) $ if $\Phi(\cdot)=\sum_l K_{l} (\cdot) K_{l}^\dagger$ is a real operation with superscript $T$ denoting transpose.

(M4) Convexity: $M(\sum _{j}p_{j}\rho_{j})\leq \sum_{j}p_{j}M(\rho_{j})$ for any probability distribution $\{p_{j}\}$ and quantum states $\{\rho_j\}$.

(M5) Additivity for direct sum states: $ M[p\rho_{1}\oplus (1-p)\rho_{2}]=pM(\rho_{1})+(1-p)M(\rho_{2}),$
where $p\in[0,1]$, $\rho_{1}$, $\rho_{2}$ are arbitrary quantum states.

Note that (M2), (M3) and (M4) together imply (M5) \cite{S.Xue2021}. Additionally, (M3) and (M4) can be derived from (M2) and (M5) \cite{S.Xue2021}. So to construct an imaginarity measure $M$, one needs to verify that it satisfies either conditions (M1), (M2) and (M5), or conditions (M1), (M2), (M3) and (M4).

There are some commonly used imaginarity measures. For example, the trace norm imaginarity
\cite{A.H.2018} is
\begin{equation}\label{eq tr im def}
 M_{\mathrm{tr}}(\rho)=\min_{\sigma\in {\cal R}} ||\rho-\sigma||_{\mathrm{tr}}  =\frac{1}{2}||\rho-\rho^*||_{\mathrm{tr}},
\end{equation}
with the trace norm $||A||_{\mathrm{tr}}=\mathrm{Tr}\sqrt{A^\dagger A}$ and the corresponding optimal real state is ${\rm Re}(\rho)$ \cite{A.H.2018}.
The relative entropy imaginarity \cite{S.N.X.2021} is
\begin{equation}\label{eq rel im def}
M_{\mathrm{rel}}(\rho)=\min_{\sigma\in {\cal R}} S(\rho || \sigma)=S(\mathrm{Re}(\rho))-S(\rho),
\end{equation}
where $S(\rho)=-\mathrm{Tr}[\rho \mathrm{log}\rho]$ is the von Neumann entropy, and $S(\rho || \sigma) =\mathrm{Tr}(\rho\log \rho)-\mathrm{Tr}(\rho\log \sigma)$.
The corresponding optimal real state in Eq. (\ref{eq rel im def}) is also ${\rm Re}(\rho)$.
The geometric imaginarity \cite{T.V.K2023} is
\begin{equation}\label{eq g im def}
M_{\mathrm{g}}(\rho)=1-\max_{\sigma\in {\cal R}} F^2(\rho,\sigma)=\frac{1}{2}(1-F(\rho,\rho^*)),
\end{equation}
with fidelity \cite{Nielsen}
\begin{equation}\label{eq fidelity}
F(\rho,\sigma)=\mathrm{Tr} \sqrt{\sqrt{\rho}\sigma\sqrt{\rho}}.
\end{equation}
The trace norm imaginarity,  relative entropy imaginarity and geometric imaginarity are all defined in terms of optimization. Fortunately, the optimizations are all solved and the analytical formulae are provided.

\section{A new method to construct imaginarity measure}
For any quantum states $\rho$ and $\sigma$ in a $d$-dimensional system,
suppose $D$ defined on any pair of states $ (\rho, \sigma)$ is a real function  satisfying
\begin{enumerate}[(1)]
 \item (Nonnegativity) $D(\rho,  \mathrm{Re}(\rho))\geq 0$ and $D(\rho,  \mathrm{Re}(\rho))=0$ if and only if $\rho= \mathrm{Re}(\rho)$.
\item (Contractive under CPTP maps)  $D(\Phi(\rho), \Phi(\sigma)) \leq D(\rho,\sigma)$ for arbitrary completely positive trace-preserving (CPTP) maps $\Phi$.
\item (Additive under direct sum states) $D(p_1\rho_1\oplus p_2\rho_2, p_1\sigma_1\oplus p_2\sigma_2) = p_1 D(\rho_1, \sigma_1)+ p_2 D(\rho_2, \sigma_2)$ for probability distributions $\{p_i\}$ and density operators $\{\rho_i\}$ and $\{\sigma_i\}$ supporting on the same subspace.
\end{enumerate}
Here the function $D$ is not necessarily a metric, because we do not demand that the function $D$ is symmetric about two components or obeys the triangle inequality.
Based on such function $D$, for any quantum state $\rho$, we define
    \begin{equation}\label{def rd}
M(\rho)=D(\rho, \mathrm{Re}(\rho)).
\end{equation}

\begin{thm}\label{th mea d}
$M(\rho)$ is a well-defined imaginarity measure.
\end{thm}

We leave the proof in \hyperlink{app:A}{Appendix A}. Theorem \ref{th mea d} provides a new method to construct imaginarity measures by the real part states.
One advantage of this kind of imaginarity measures is the absence of optimization, so it may be convenient for calculation. Furthermore, the real part state $\mathrm{Re}(\rho)$ has less parameters compared with general quantum state.

Here we give a summary about the existing imaginarity measures with respect to the real part state $\mathrm{Re}(\rho)$ in Table \ref{table}, where the Tsallis relative entropy is defined as $D_{\alpha}(\rho||\sigma)=\frac{\mathrm{Tr}\left(\rho^{\alpha}\sigma^{1-\alpha}\right)-1}{\alpha-1}$ with $\alpha\in (0,1)$, and the generalized quantum Jensen-Shannon divergence (GQJSD) is defined as $J_{\alpha}(\rho,\sigma)=\frac{1}{2}\left(D_\alpha\left(\rho||\frac{\rho+\sigma}{2}\right)+D_\alpha\left(\sigma||\frac{\rho+\sigma}{2}\right)\right)$ \cite{P.Tian2025}. We can see that these
imaginarity measures are factually all specific cases of Theorem \ref{th mea d} and the corresponding functions $D$  are also listed.

\begin{table}[h]
\centering
\small
\setlength{\extrarowheight}{12pt}
\caption{A summary of the imaginarity measures with respect to the  real part state $\mathrm{Re}(\rho)$.}
\label{tab:measures}
\begin{tabular}{|c|c|c|}
\hline
 &  Imaginarity measures  & Function $D(\rho, \sigma)$ \\  \hline
Trace norm, Robustness  &$ ||\rho-\mathrm{Re}(\rho)||_{\mathrm{tr}}$ \cite{A.H.2018} \cite{WuKD20212} &  $  ||\rho-\sigma||_{\mathrm{tr}}$ \\        
\hline      
Relative entropy &
$S(\rho||\mathrm{Re}(\rho))$
\cite{S.N.X.2021} &
 $S(\rho||\sigma)$ \\        
 \hline
Tsallis relative entropy  & $D_{\alpha}\left(\rho||\mathrm{Re}(\rho)\right)$ \cite{J.W.X.2024} \cite{P.Tian2025} & $\frac{\mathrm{Tr}\left(\rho^{\alpha}\sigma^{1-\alpha}\right)-1}{\alpha-1}$ \\  
\hline
{The GQJSD}  & $J_{\alpha}(\rho,\mathrm{Re}(\rho))$ \cite{P.Tian2025} & $\frac{1}{2}\left(D_{\alpha}\left(\rho||\frac{\rho+\sigma}{2}\right)+D_{\alpha}\left(\sigma||\frac{\rho+\sigma}{2}\right)\right)$ \\  
\hline
\end{tabular}
\label{table}
\end{table}

Now  for any quantum states $\rho$ and $\sigma$,
if we specify the function $D$ as
$D(\rho,\sigma)=1-F(\rho,\sigma)$, one can verify that this function meets the three conditions required.
Subsequently, we derive the imaginarity measure
 \begin{equation}\label{eq def mre}
        M_{\mathrm{Re}}(\rho)=1-F(\rho,\mathrm{Re}(\rho)).
    \end{equation}
In \hyperlink{app:B}{Appendix B}, we shall prove that the imaginarity measure
$M_{\mathrm{Re}}(\rho)$ satisfies the following properties.

\begin{thm}\label{th properties of mre}
     $M_{\mathrm{Re}}(\rho)$ is a well-defined imaginarity measure and satisfies the following properties.
     \begin{enumerate}[(1)]
         \item $ 0\leq M_{\mathrm{Re}}(\rho)\leq 1-\frac{\sqrt{2}}{2}$. $ M_{\mathrm{Re}}(\rho)=0$ if and only if $\rho$ is real. The upper bound $1-\frac{\sqrt{2}}{2}$ is attainable when $|\psi\rangle$ is the maximally imaginary state $|+ \rm i\rangle=\frac{|0\rangle + {\rm i} |1\rangle}{\sqrt{2}}$.
         \item For any pure state  $|\psi\rangle$,
         \begin{equation}  \label{eq Mre pure}
             M_{\mathrm{Re}}(|\psi\rangle)=1- \frac{1}{\sqrt{2}}\sqrt{{1+ |\langle \psi|\psi^*\rangle|^2}}.
         \end{equation}
         \item For any pure state $|\psi \rangle$, the geometric imaginarity $M_{\mathrm{g}}(\rho)$ and the imaginarity measure $M_{\mathrm{Re}}(\rho)$ satisfy
         \begin{equation}  \label{eq relation of MRe and Mg}
             M_{\mathrm{Re}}(|\psi\rangle)=1- \frac{1}{\sqrt{2}}\sqrt{{1+(1-2M_g(|\psi\rangle))^2}}.
         \end{equation}
         \item
       For any set $\{\Pi_i\}_{i=0}^{d-1}$ of rank-$1$ projection operators with real entries,
        \begin{eqnarray}
       M_{{\rm Re}}\left( \sum_i \Pi_i \rho \Pi_i\right) = \sum\limits_{i} p_i  M_{{\rm Re}}  \left(\rho_{i}\right)
    \end{eqnarray}
    with  $\rho_{i}=\frac{1}{p_i}\Pi_i \rho \Pi_i$, $p_i=\mathrm{Tr}(\Pi_i \rho \Pi_i)$, $i=0,1,\cdots, d-1$.
     \end{enumerate}
\end{thm}

In the case of single qubit, it is straightforward to derive the analytical expression of imaginarity measure $M_{\mathrm{Re}}$.

\begin{thm} \label{th:4}
    For any qubit state $  \rho=\frac{1}{2} ({I}+\boldsymbol{r}\cdot\boldsymbol{\sigma})$
   with $\boldsymbol{r}=(r_x, r_y, r_z)\in \mathbb{R}^3$, $|\boldsymbol{r}|=\sqrt{r_x^2+ r_y^2+ r_z^2}\leq 1$, $\boldsymbol{\sigma}=(\sigma_x,\sigma_y,\sigma_z)$ with three Pauli operators, $\sigma_x=|0\rangle\langle 1| + |1\rangle\langle 0|$, $\sigma_y={\rm{i} }(|1\rangle\langle 0|-|0\rangle\langle 1|)$, $\sigma_z=|0\rangle\langle 0| - |1\rangle\langle 1|$, the analytical expression of imaginarity measure $M_{\mathrm{Re}}$ is
    \begin{equation}\label{eq mre 2}
        M_{\mathrm{Re}}(\rho)=1-\Bigg[\frac{\sqrt{(1-|\boldsymbol{r}|^2)(1-|\boldsymbol{r}|^2+r_y^2)}+1+|\boldsymbol{r}|^2-r_y^2}{2}\Bigg]^\frac{1}{2}.
    \end{equation}
\end{thm}

We leave the proof in \hyperlink{app:C}{Appendix C}.
From Theorem \ref{th:4} we see that the imaginarity of $\rho$ increases with
the coefficient $|r_y|$ and the latter is confined by the length of the Bloch vector $|\boldsymbol{r}|$
(See FIG. \ref{fig:Visualization of the ternary function M_Re}).

\begin{figure}[htbp]
\centering

\begin{minipage}{0.48\textwidth}
    \centering
    \includegraphics[width=\linewidth]{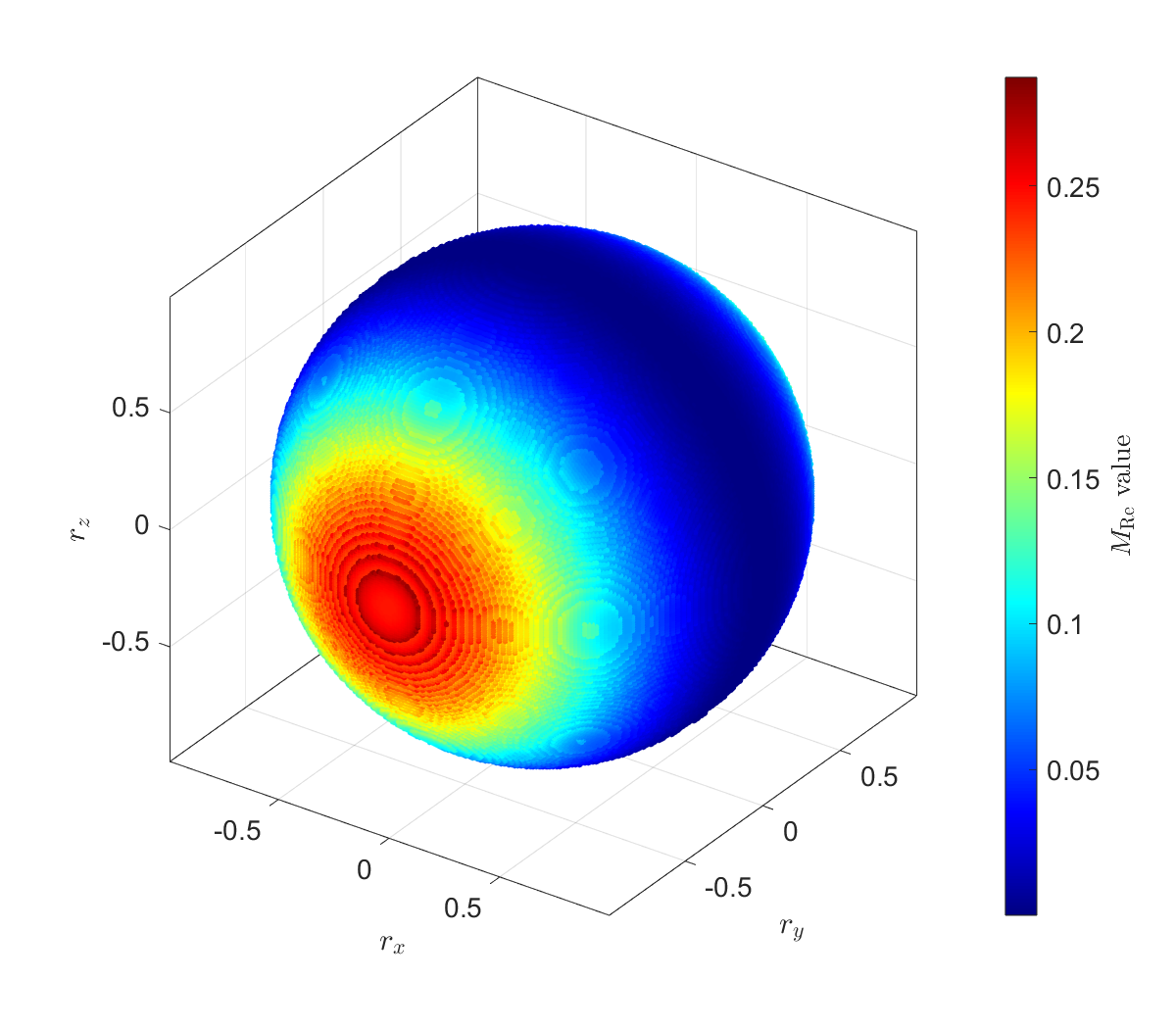}
    \\ (a)
\end{minipage}
\hfill
\begin{minipage}{0.48\textwidth}
    \centering
    \includegraphics[width=\linewidth]{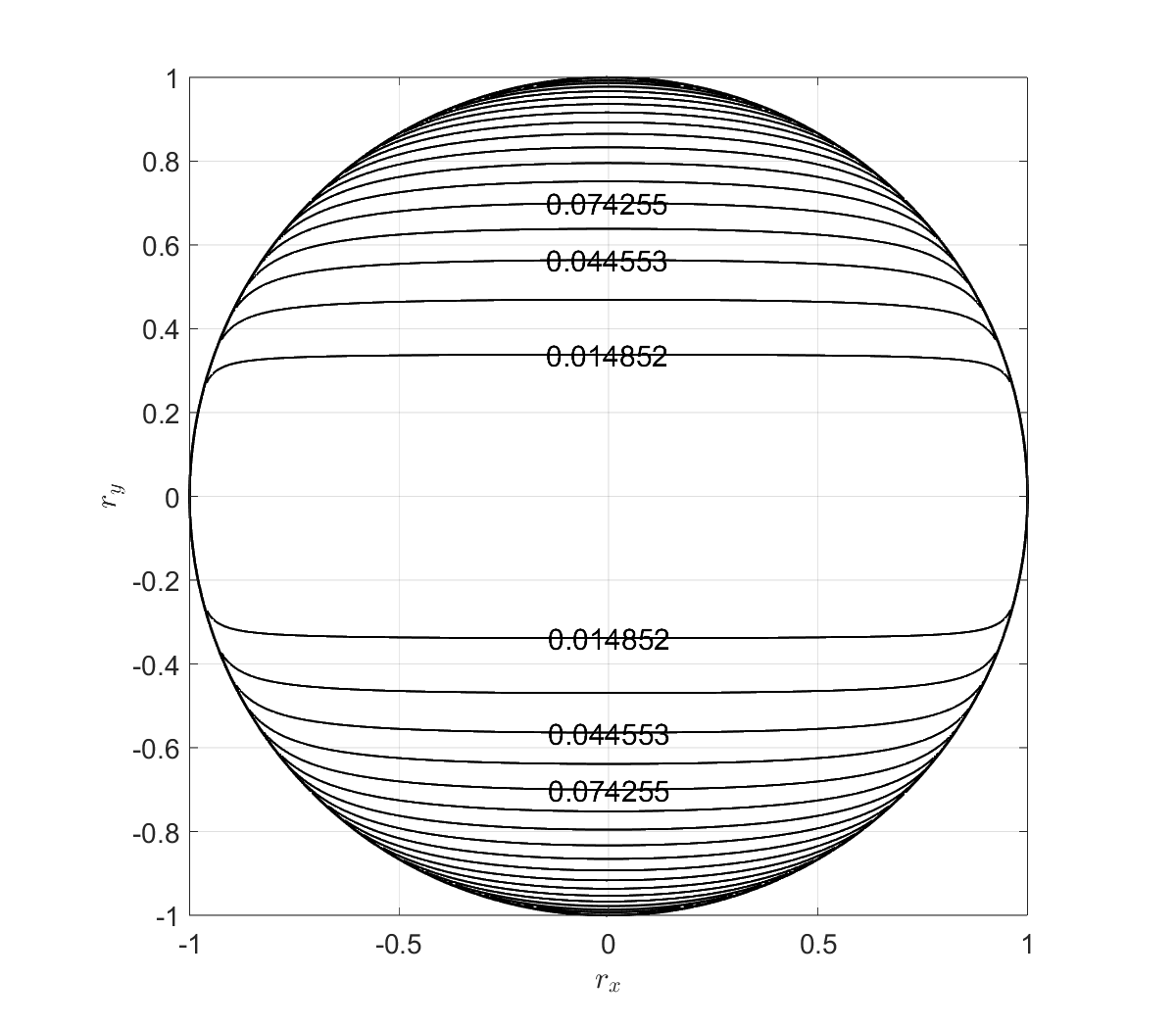}
    \\ (b)
\end{minipage}

\caption{(Color online)  The subfigure (a) is the visualization of the ternary function $M_{\mathrm{Re}}$ in Eq. (\ref{eq mre 2}) in qubit systems. The subfigure (b) is the contour plot of $M_{\mathrm{Re}}$ with $r_z=0$.}
\label{fig:Visualization of the ternary function M_Re}
\end{figure}

Obviously the geometric imaginarity $M_{\mathrm{g}}$ in Eq. (\ref{eq g im def}) and the imaginarity measure $ M_{\mathrm{Re}}$ in Eq. (\ref{eq def mre}) are both defined in form of fidelity. In fact, these two imaginarity measures are different but closely related.

\begin{thm}  \label{theo:geo measure}
    For any quantum state $\rho$, the geometric imaginarity $ M_{\mathrm{g}}(\rho)$ and the imaginarity measure $   M_{\mathrm{Re}}(\rho)$ satisfy
    \begin{equation}\label{eq mr mg}
        M_{\mathrm{Re}}(\rho) \leq M_{\mathrm{g}}(\rho) \leq 1-[1- M_{\mathrm{Re}}(\rho)]^2.
    \end{equation}
    The equality $ M_{\mathrm{Re}}(\rho) = M_{\mathrm{g}}(\rho)$ holds if and only if $\rho$ is a real state. And $M_{\mathrm{g}}(\rho)$ reaches the upper bound  $1-[1- M_{\mathrm{Re}}(\rho)]^2$ if and only if the optimal real state of $\rho$ with respect  to the geometric imaginarity $ M_{\mathrm{g}}(\rho)$ is the real part state ${\mathrm{Re}}(\rho)$.
\end{thm}

The proof is provided in \hyperlink{app:D}{Appendix D}.
Similarly, for the Tsallis relative entropy imaginarity $M_{T,\frac{1}{2}}(\rho)=1-A(\rho,\rho^*)$ \cite{J.W.X.2024} with $A(\rho,\sigma)=\mathrm{Tr}(\sqrt{\rho}\sqrt{\sigma})$ denoting the quantum affinity,
employing the relation  $A(\rho,\sigma) \leq F(\rho,\sigma)$ \cite{S.L.Luo2004} and the formula of the geometric imaginarity in Eq. (\ref{eq g im def}), we have $2M_{\mathrm{g}}(\rho) \leq M_{T,\frac{1}{2}}(\rho)$. Hence, we can derive the relation between the imaginarity measure $   M_{\mathrm{Re}}$ and Tsallis relative entropy imaginarity as follows.

\begin{cor}
     For any quantum state $\rho$,  the imaginarity measure $M_{\mathrm{Re}}(\rho)$ and the Tsallis relative entropy imaginarity $ M_{T,\frac{1}{2}}(\rho)$ satisfy
    \begin{eqnarray*}
        M_{\mathrm{Re}}(\rho) \leq \frac{1}{2}M_{T,\frac{1}{2}}(\rho),
    \end{eqnarray*}
    and the equality holds if and only if $\rho$ is a real state.
\end{cor}

Now we define
\begin{equation}\label{eq mg-2}
M_{\mathrm{g}}^\prime (\rho)=1-\max_{\sigma\in {\cal R}} F(\rho,\sigma)
\end{equation}
by removing the square from the fidelity $F$ in the definition of geometric imaginarity. One can check that $M_{\mathrm{g}}^\prime (\rho)$ satisfies conditions (M1) and (M2). As a quantifier of the imaginarity,
employing the formula in Eq. (\ref{eq g im def}), the analytical expression of $M_{\mathrm{g}}^\prime (\rho)$ is $ M_{\mathrm{g}}^\prime (\rho)= 1- \left[\frac{1+F(\rho,\rho^*)}{2}\right]^{\frac{1}{2}}$.
It is related to the imaginarity measure $M_{\mathrm{Re}}$, and the relation is given by
\begin{equation}
     M_{\mathrm{g}}^\prime (\rho) \leq  M_{\mathrm{Re}}(\rho).
\end{equation}
Now we take an example to compare these imaginarity measures illustratively.

\begin{example}   \label{ex Mg MRe MT}
\rm{Let $|+{\rm i}\rangle=\frac{1}{\sqrt{2}}(|0\rangle + {\rm i} |1\rangle )$, consider the mixed state
   \begin{equation}\label{eq ex rho-2-o}
        \rho=p |+{\rm i}\rangle\langle +{\rm i}| +(1-p)\frac{I}{2}
 \end{equation}
   with $0\leq p\leq 1$.
   It can be easily calculated that the geometric imaginarity and the normalized Tsallis relative entropy imaginarity $ M^\prime_{T,\frac{1}{2}}(\rho)=\frac{1}{2}M_{T,\frac{1}{2}}(\rho)$ of $\rho$  coincide and equal to
   \begin{equation}  \label{eq mg rho}
        M^\prime_{T,\frac{1}{2}}(\rho)=  M_{\mathrm{g}}(\rho)=\frac{1-\sqrt{1-p^2}}{2},
   \end{equation}
   the imaginarity measure $M_{\mathrm{Re}}$ and the quantifier of the imaginarity $M_{\mathrm{g}}'$ of $\rho$ coincide and equal to
\begin{equation}   \label{eq mre rho}
     M_{\mathrm{Re}}(\rho)=M_{\mathrm{g}}'(\rho)=1-\sqrt{\frac{\sqrt{1-p^2}+1}{2}}.
\end{equation}
   The imaginarity measure $M_{\mathrm{Re}}(M_{\mathrm{g}}')$ and $M_{\mathrm{g}} ( M^\prime_{T,\frac{1}{2}})$ increase with the parameter $p$ (See FIG. \ref{fig:The graphical representations of the functions Mg and MRe}).
When $p=0$, the quantum state in Eq. \eqref{eq ex rho-2-o} becomes maximally mixed state $I/2$ which is real state. In this case, the imaginarity measure $M_{\mathrm{Re}}(M_{\mathrm{g}}')$ and $M_{\mathrm{g}} ( M^\prime_{T,\frac{1}{2}})$ vanish. When $p=1$, the quantum state in Eq. \eqref{eq ex rho-2-o} becomes maximally imaginary state $|\rm +i\rangle$. In this case, the imaginarity measure $M_{\mathrm{Re}}(M_{\mathrm{g}}')$ and $M_{\mathrm{g}} ( M^\prime_{T,\frac{1}{2}})$ reach their maximum values.
Furthermore, one can check that  the geometric imaginarity $M_{\mathrm{g}}$ reaches its upper bound $1-[1- M_{\mathrm{Re}}(\rho)]^2$ for $0\leq p \leq 1$.
}
\end{example}

\begin{figure}[htb]
\centering
    \includegraphics[width=0.4\textwidth]{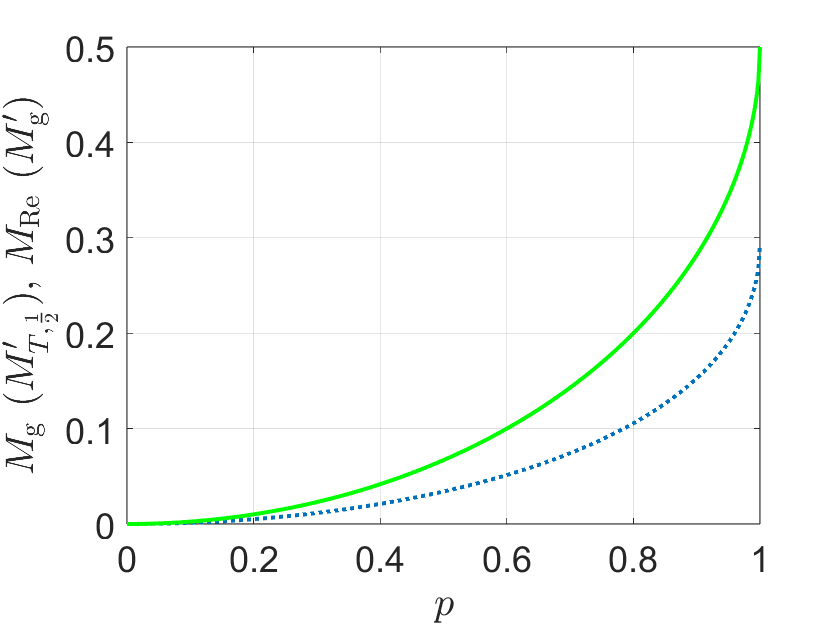}
    \caption{(Color Online) The comparison of the imaginarity measures $M_{\mathrm{g}}\ ( M^\prime_{T,\frac{1}{2}})$ and $M_{\mathrm{Re}}\ (M_{\mathrm{g}}')$ for qubit states in Eq. (\ref{eq ex rho-2-o}), with $M_{\mathrm{g}}\ (M^\prime_{T,\frac{1}{2}})$ shown as the solid green curve and $M_{\mathrm{Re}}\ (M_{\mathrm{g}}')$ as the dashed blue curve. }
    \label{fig:The graphical representations of the functions Mg and MRe}

\end{figure}

As we have mentioned earlier, the geometric imaginarity as well as the trace norm imaginarity and relative entropy imaginarity are all defined in form of optimization initially. The optimalizations for the trace norm imaginarity and relative entropy imaginarity are attained by the real part state $\mathrm{Re}(\rho)$. However the optimal real state for the geometric imaginarity is still unknown. Next we explore the optimal real state for the geometric imaginarity in qubit systems.

 \begin{thm}\label{th mg real}
      For any qubit state 
         \begin{equation*}
    \rho=\frac{1}{2} (I+\boldsymbol{r}\cdot\boldsymbol{\sigma})=
      \frac{1}{2}
    \left( \begin{array}{cc}
     1+ r_z& r_x-{\rm i} r_y\\
    r_x+{\rm i} r_y & 1-r_z
   \end{array}
  \right)
\end{equation*}
with $\boldsymbol{r}=(r_x, r_y, r_z)\in \mathbb{R}^3$, $|\boldsymbol{r}|=\sqrt{r_x^2+ r_y^2+ r_z^2}\leq 1$,
      \begin{enumerate}[(1)]
          \item the geometric imaginarity is
          \begin{equation}\label{eq mg 2}
              M_{\mathrm{g}}(\rho)=\frac{1}{2}(1-\sqrt{1-r_y^2});
          \end{equation}
          \item its optimal real state is
    \begin{equation}\label{eq mg rel}
       \sigma_\star = \frac{1}{2}(I + x_0  \sigma_x + z_0  \sigma_z),
    \end{equation}
        where $x_0=\frac{r_x}{\sqrt{1-r_y^2}}, z_0=\frac{r_z}{\sqrt{1-r_y^2}}$.
      \end{enumerate}
 \end{thm}
The analytical formula in Eq. (\ref{eq mg 2}) is obtained in  \cite{T.T.Xia2024}. For the proof of the optimal real state in Eq. (\ref{eq mg rel}), we leave it in the \hyperlink{app:E}{Appendix E}.
Theorem \ref{th mg real} tells us that the optimal real state in accordance with the geometric imaginarity is generally not the real part state. It coincides with the real part state if and only if $r_x=r_z=0$ or $r_y=0$. That is also the reason why the state given in Eq. \eqref{eq ex rho-2-o} enable geometric imaginarity to reach its upper bound in Theorem \ref{theo:geo measure}.

Additionally, we find that the imaginarity measure $M_{\mathrm{Re}}$ is bounded by the trace norm imaginarity as follows, where we denote the purity $P(\rho)={\rm Tr}(\rho^2)$.

\begin{thm}  \label{thm:trace norm}
    For any quantum state $\rho$,  the imaginarity measure $M_{\mathrm{Re}}$ and the trace norm imaginarity satisfy
    \begin{eqnarray}   \label{eq tradeoff}
        1-\left[\frac{1}{d}+\frac{1}{d}\sqrt{(d-1)^2-(d-1)M_{\mathrm{tr}}(\rho)^2}\right]^{\frac{1}{2}} \leq M_{\mathrm{Re}}(\rho) \leq 1-\big(P(\rho)-M_{\mathrm{tr}}(\rho)^2\big)^{\frac{1}{2}}.
    \end{eqnarray}
\end{thm}

We leave the proof in \hyperlink{app:F}{Appendix F}.
The second inequality in Theorem \ref{thm:trace norm}  indicates a relation among the imaginarity measure $ M_{\mathrm{Re}}$, the trace norm imaginarity $ M_{\mathrm{tr}}$ and the purity
\begin{eqnarray*}
       [1- M_{\mathrm{Re}}(\rho) ]^2 + M_{\mathrm{tr}}(\rho)^2 \geq P(\rho).
    \end{eqnarray*}
This relation shows for any given quantum state $\rho$, the imaginarity measure $ M_{\mathrm{Re}}$ increases with the trace norm imaginarity $ M_{\mathrm{tr}}$. However the purity limits the sum of $ [1- M_{\mathrm{Re}}(\rho) ]^2$ and $M_{\mathrm{tr}}(\rho)^2$.
Next, we illustrate the evaluation of the imaginarity measure $M_{\mathrm{Re}}$ with a specific example, using the trace norm imaginarity $M_{\mathrm{tr}}$.

\begin{example}
    \rm {Turning again to the mixed state given by Eq. \eqref{eq ex rho-2-o} in Example \ref{ex Mg MRe MT}.
    It can be easily calculated that the trace norm imaginarity is $ M_{\mathrm{tr}}(\rho)=p$ and the purity of $\rho$ is $P(\rho)=\frac{1+p^2}{2}$.
In consideration of the imaginarity measure $M_{\mathrm{Re}}(\rho)$ in Eq. (\ref{eq mre rho}), we find the imaginarity measure $M_{\mathrm{Re}}$, trace norm imaginarity $M_{\mathrm{tr}}$ and purity $P$ increase with the parameter $p$.
When $p=0$, the quantum state in Eq. (\ref{eq ex rho-2-o}) becomes maximally mixed state $I/2$. In this case, both $M_{\mathrm{Re}}(\rho)$ and  $M_{\mathrm{tr}}(\rho)$ vanish. Meanwhile the purity $ P(\rho)$ reaches its minimum value 1/2.
When $p=1$, the quantum state in Eq. (\ref{eq ex rho-2-o}) becomes the maximally imaginary state $|\rm +i\rangle$. In this case, both  $M_{\mathrm{Re}}(\rho)$ and $M_{\mathrm{tr}}(\rho)$ reach their maximum values. The visualization of the relation in Eq. (\ref{eq tradeoff}) is presented in FIG. \ref{fig:The comparison of the imaginarity measures MRe ,Mtr and the purity for qubit}, from which we see the imaginarity measure $M_{\mathrm{Re}}(\rho)$ coincides with its lower bound. So the lower bound in Eq. (\ref{eq tradeoff}) is tight.
}
\end{example}

\begin{figure}[ht]
    \centering
    \includegraphics[width=0.4\textwidth]{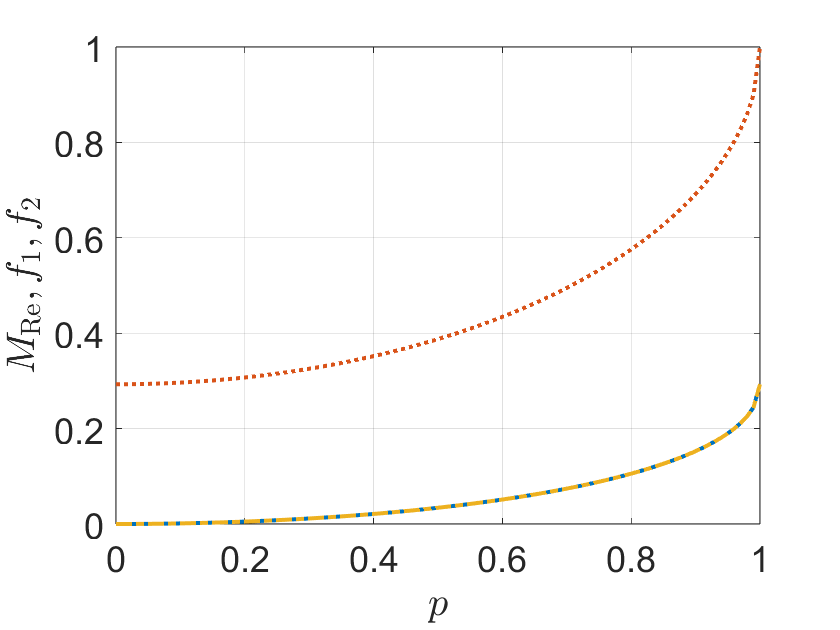}
    \caption{(Color Online) The visualization of the bounds for $M_{\mathrm{Re}}$ in Theorem \ref{thm:trace norm} for the qubit state in Eq. \eqref{eq ex rho-2-o}. The red dashed curve above represents $f_1=1-\big(P(\rho)-M_{\mathrm{tr}}(\rho)^2\big)^{\frac{1}{2}}$ corresponding to the upper bound in Eq. \eqref{eq tradeoff}.
    The yellow dashed curve below represents
    $f_2=1-\left[\frac{1}{d}+\frac{1}{d}\sqrt{(d-1)^2-(d-1)M_{\mathrm{tr}}(\rho)^2}\right]^{\frac{1}{2}}$ corresponding to the lower bound in Eq. \eqref{eq tradeoff}.
   The blue solid curve represents $M_{\mathrm{Re}}$ and it coincides with the lower bound  $f_2$. }
    \label{fig:The comparison of the imaginarity measures MRe ,Mtr and the purity for qubit}
\end{figure}

\section{The complementarity relations of the imaginary measure ${M_{\mathrm{Re}}}$ }

In this section we show the complementarity relation of imaginary measure $M_{\mathrm{Re}}$ under a complete set of MUBs. First in any qubit systems, suppose $B_1$, $B_2$, and $B_3$ are three orthonormal bases composed by the eigenvectors of three Pauli operators,
\begin{eqnarray*}
    \begin{array}{rcl}
       B_1&=& \{|0\rangle,\ |1\rangle\},\\
       B_2&=& \{\frac{1}{\sqrt{2}}(|0\rangle + {\rm i}|1\rangle),\ \frac{1}{\sqrt{2}}(|0\rangle- {\rm i}|1\rangle)\},\\
       B_3&=&\{\frac{1}{\sqrt{2}}(|0\rangle + |1\rangle),\ -\frac{\rm i}{\sqrt{2}}(|0\rangle- |1\rangle)\}.
    \end{array}
    \end{eqnarray*}
    These three orthonormal bases $B_1$, $B_2$ and $B_3$ are mutually unbiased. In \hyperlink{app:G}{Appendix G}, we shall prove the complementarity relation of imaginary measure $M_{\mathrm{Re}}$ under the set of $\{B_i\}_{i=1,2,3}$.

\begin{thm}  \label{th comple of qubit}
   For any qubit state $\rho$, the imaginarity measure $M_{\mathrm{Re}}$ under mutually unbiased bases $B_1$, $B_2$ and $B_3$ satisfies
    \begin{eqnarray}   \label{eq trade-off for qubit}
        \left(1-M_{\mathrm{Re}}^{B_1}\right)^2+\left(1-M_{\mathrm{Re}}^{B_2}\right)^2+\left(1-M_{\mathrm{Re}}^{B_3}\right)^2 \geq \frac{\sqrt{(1-|\boldsymbol{r}|^2)(3-2|\boldsymbol{r}|^2)}+3+2|\boldsymbol{r}|^2}{2},
    \end{eqnarray}
    with $|\boldsymbol{r}|$ the length of the  Bloch vector of $\rho$. And the equality holds if and only if $\rho$ is pure.
\end{thm}

Theorem \ref{th comple of qubit} demonstrates that the imaginarity of quantum state under a complete set of MUBs is confined by the length of the  Bloch vector of $\rho$ (See FIG. \ref{fig:The complement for 2}).
The relation in Eq. (\ref{eq trade-off for qubit}) turns into the equality for qubit pure states. For any high-dimensional pure state, there exists a real orthogonal matrix $O$ that transforms it to a pure qubit state \cite{A.H.2018} \cite{WuKD20212}. Since real orthogonal transformations preserve the imaginarity measure of quantum states, we have the following complete complementarity relation.

\begin{figure}[ht]
    \centering
    \includegraphics[width=0.4\textwidth]{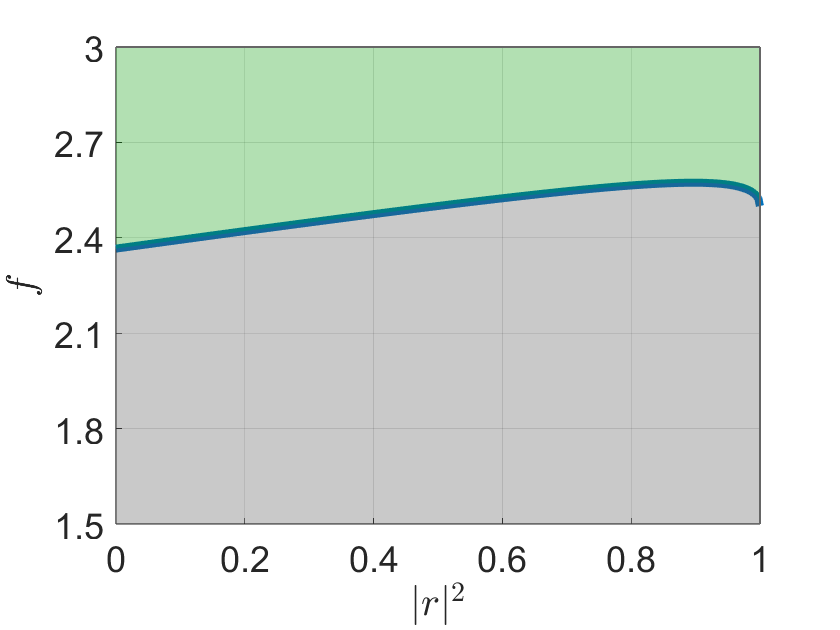}
    \caption{(Color Online) The green part above is the region available for the imaginarity measure $M_{\mathrm{Re}}$. Here the vertical axis $f= \left(1-M_{\mathrm{Re}}^{B_1}\right)^2+\left(1-M_{\mathrm{Re}}^{B_2}\right)^2+\left(1-M_{\mathrm{Re}}^{B_3}\right)^2$ corresponds to the left-hand side of Eq. (\ref{eq trade-off for qubit}), while the blue curve plots the function $\frac{\sqrt{(1-|\boldsymbol{r}|^2)(3-2|\boldsymbol{r}|^2)}+3+2|\boldsymbol{r}|^2}{2}$ with respect to $|r|^2$, corresponding to the right-hand side. }
    \label{fig:The complement for 2}
\end{figure}

\begin{cor}   \label{coll qutrit pure compl}
    For any $d$-dimensional pure state, there exists three sets of orthonormal bases $Y_i$ ($i=1,2,3$) such that
    the imaginarity measure $M_{\mathrm{Re}}$ satisfies
    \begin{eqnarray}\label{eq com rel pure}
        \left(1-M_{\mathrm{Re}}^{Y_1}\right)^2+\left(1-M_{\mathrm{Re}}^{Y_2}\right)^2+\left(1-M_{\mathrm{Re}}^{Y_3}\right)^2 =\frac{5}{2}.
    \end{eqnarray}
   \end{cor}

   It should be noted that these orthonormal bases $\{Y_i\}$ are state-dependent as the  real orthogonal matrix $O$ transforming high dimensional pure state to qubit pure state is state-dependent. Here we show an example for this complemantarity relation in Eq. (\ref{eq com rel pure}) in three dimensional systems.

\begin{example}
    \rm{Consider the qutrit pure state  $|\phi\rangle=\frac{1}{\sqrt{3}} (|0\rangle+  |1\rangle + {\rm i} |2\rangle )$,
 We specify
    \begin{eqnarray}
    \begin{array}{rcl}
       Y_1&=& \{|0\rangle,\ |1\rangle,\ |2\rangle\},\\
       Y_2&=& \{\frac{\sqrt{2}+2}{4}|0\rangle + \frac{\sqrt{2}-2}{4}|1\rangle + \frac{\mathrm{i}}{2}|2\rangle,\ \frac{\sqrt{2}-2}{4}|0\rangle + \frac{\sqrt{2}+2}{4}|1\rangle + \frac{\mathrm{i}}{2}|2\rangle,\ \frac{1}{2}|0\rangle + \frac{1}{2}|1\rangle - \frac{\mathrm{i}}{\sqrt{2}}|2\rangle\},\\
       Y_3&=&\{\frac{\sqrt{2}+2}{4}|0\rangle + \frac{\sqrt{2}-2}{4}|1\rangle + \frac{\mathrm{1}}{2}|2\rangle,\ \frac{\sqrt{2}-2}{4}|0\rangle + \frac{\sqrt{2}+2}{4}|1\rangle + \frac{\mathrm{1}}{2}|2\rangle,\ -\frac{\mathrm{i}}{2}|0\rangle - \frac{\mathrm{i}}{2}|1\rangle + \frac{\mathrm{i}}{\sqrt{2}}|2\rangle\}.
    \end{array}
    \end{eqnarray}
 Under these three orthonormal bases  $Y_1$, $Y_2$ and $Y_3$, the imaginaries of $  |\phi\rangle$ are $M_{\mathrm{Re}}^{Y_1}(|\phi\rangle)=1-\frac{\sqrt{5}}{3}$, $M_{\mathrm{Re}}^{Y_2}(|\phi\rangle)=0$ and $M_{\mathrm{Re}}^{Y_3}(|\phi\rangle)=1-\frac{\sqrt{34}}{6}$ respectively. This gives rise to the complementarity relation in Eq. (\ref{eq com rel pure}).
 }
\end{example}

In qutrit systems, a complete set of MUBs can be constructed as follows:
\begin{eqnarray}
    \begin{array}{rcl}
       Z_1&=& \{|0\rangle,\ |1\rangle,\ |2\rangle\},\\
       Z_2&=& \{\frac{1}{\sqrt{3}}(|0\rangle + |1\rangle + |2\rangle),\ \frac{1}{\sqrt{3}}(|0\rangle+ \omega |1\rangle + \omega^2 |2\rangle),\ \frac{1}{\sqrt{3}}(|0\rangle+ \omega^2 |1\rangle + \omega |2\rangle)\},\\
       Z_3&=&\{\frac{1}{\sqrt{3}}(|0\rangle + \omega |1\rangle + \omega |2\rangle),\ \frac{1}{\sqrt{3}}(|0\rangle+ \omega^2 |1\rangle +  |2\rangle),\ \frac{1}{\sqrt{3}}(|0\rangle+ |1\rangle + \omega^2 |2\rangle)\},\\
       Z_4&=&\{\frac{1}{\sqrt{3}}(|0\rangle + \omega^2 |1\rangle + \omega^2 |2\rangle),\ \frac{1}{\sqrt{3}}(|0\rangle+  |1\rangle + \omega |2\rangle),\ \frac{1}{\sqrt{3}}(|0\rangle+ \omega|1\rangle + |2\rangle)\},
    \end{array}
    \end{eqnarray}
    where $\omega=e^{\frac{2 \pi \mathrm{i}}{3}}$ is the complex $3$rd root of unity. In \hyperlink{app:H}{Appendix H}, we shall prove the complementarity relation of $M_{\mathrm{Re}}$ under the complete set of MUBs $\{Z_k\}_{k=1,2,3,4}$.

\begin{thm}     \label{th qutrit pure trade off mutually bases}
    For any qutrit state $\rho$, the imaginarity measure $M_{\mathrm{Re}}$ under the set of MUBs $\{Z_k\}_{k=1,2,3,4}$ satisfies
    \begin{eqnarray}   \label{eq compl qutrit pure state under complete mutually}
  \sum_{k=1}^4  \left(1-M_{\mathrm{Re}}^{Z_k}\right)^2 > \frac{31}{14}P(\rho).
    \end{eqnarray}
\end{thm}

Theorem \ref{th qutrit pure trade off mutually bases} shows the imaginarity under a complete set of MUBs is constrained by the purity of quantum states. This complementarity relation is illustrated in FIG. \ref{fig:The complement for 3} by 2000 randomly generated states numerically, from which it follows that the left-hand side converges to the right-hand side of Eq. (\ref{eq compl qutrit pure state under complete mutually}) as the quantum state evolves toward purity.

\begin{figure}[htbp]
    \centering
    \includegraphics[width=0.4\textwidth,height=0.2\textheight, keepaspectratio]{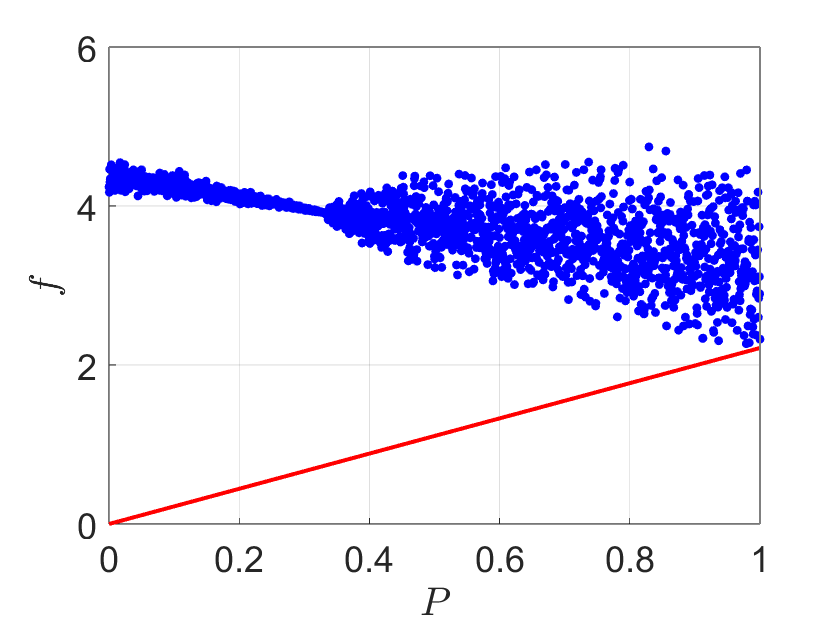}
    \caption{(Color Online) Visualization of the complementarity relation in Eq. \eqref{eq compl qutrit pure state under complete mutually} for 2000 randomly generated qutrits. The blue solid dots represent $ f=\sum\limits_{k=1}^4 \left(1-M_{\mathrm{Re}}^{Z_k}\right)^2$ corresponding to the left-hand side of Eq. \eqref{eq compl qutrit pure state under complete mutually} with respect to the purity $P$ for randomly generated qutrit states. The red line represents the function $\frac{31}{14}P$ corresponding to the right-hand side of Eq. \eqref{eq compl qutrit pure state under complete mutually}. }
    \label{fig:The complement for 3}
\end{figure}

The complementarity relations in Theorems \ref{th comple of qubit} and
\ref{th qutrit pure trade off mutually bases} tell us that for any given quantum state, the amounts of imaginarity can not be all large for the complete set of MUBs. More specifically, since the complementarity relation in Theorem \ref{th comple of qubit} is complete for pure state, so the amounts of imaginarity
under the MUBs $B_1$, $B_2$ and $B_3$ can not be all large or small.
This relation reflects the dependence of imaginarity on the measurement basis.

\section{Conclusions}
To summarize, we have developed a method to construct imaginary measures by real part state. This class of imaginary measures includes the existing imaginary measures in form of the real part state.
Specifically, we have proposed a new imaginarity measure in terms of fidelity and explored its properties. In qubit systems, we have derived the analytical expression for the imaginarity measure.
The relations between this imaginarity measure and some other imaginarity measures have been established.
The complementarity relation under a complete set of MUBs have also been established in low-dimensional systems. This work deepens the interpretation of the real part state on characterizing quantum states.
Meanwhile, it also contributes to characterizing the dependence of imaginarity on measurement bases.
\section{Acknowledgment}
M. J. Zhao thanks the center for Quantum Information,
Institute for Interdisciplinary Information Sciences,
Tsinghua University for hospitality.
This work is supported by the National Natural Science Foundation of China under grant Nos. 12171044 and 12471443, and Postgraduate Research \& Practice Innovation Program of Jiangsu Province (KYCX25\_0627).

\section{Appendix}
\appendixsection{The proof of Theorem \ref{th mea d}}

 First, the faithfulness (M1) is obvious.
Second, for any quantum state $\rho$ and real operation $\Phi$, it has
\begin{eqnarray*}
   M (\Phi(\rho))&=&D(\Phi(\rho),  {\rm Re}(\Phi(\rho)))\\
    &=& D(\Phi(\rho), \Phi( {\rm Re}(\rho)))\\
    &\leq& D(\rho,  {\rm Re}(\rho))\\
    &=&M (\rho),
\end{eqnarray*}
where the first equality is the definition in Eq. (\ref{def rd}), the second equality is the property of real operation $\Phi$, the inequality in the third line is the contractivity of the function $D$. So $M(\rho)$ satisfies the monotonicity (M2).
Third, for arbitrary probability distribution $\{p_i\}$ and density operators $\{\rho_i\}$,
\begin{eqnarray*}
 M(p_1\rho_1\oplus p_2\rho_2)&=&
D(p_1\rho_1\oplus p_2\rho_2, \mathrm{Re}(p_1\rho_1\oplus p_2\rho_2)) \\
&=& D(p_1\rho_1\oplus p_2\rho_2, p_1\mathrm{Re}(\rho_1)\oplus p_2\mathrm{Re}(\rho_2)) \\
&=&   p_1 D(\rho_1, \mathrm{Re}(\rho_1))+ p_2 D(\rho_2,  \mathrm{Re}(\rho_2))\\
&=&  p_1  M(\rho_1) +p_2  M (\rho_2),
\end{eqnarray*}
where the first equality is the definition of $M(\rho)$ in Eq. (\ref{def rd}), the second equality is the property of real part, the third equality is the additivity under direct sum states of the function $D$. So $M(\rho)$ satisfies the condition (M5). Therefore $M(\rho)$ is a well-defined imaginarity measure.

\appendixsection{The proof of  Theorem \ref{th properties of mre}}
(1) First, since $0 \leq F(\rho, \mathrm{Re}(\rho))\leq 1$ and $F(\rho, \mathrm{Re}(\rho))=1$ holds if and only if $\rho=\mathrm{Re}(\rho)$ \cite{Nielsen}, this implies that $M_{\mathrm{Re}}(\rho) \geq 0$ and $M_{\mathrm{Re}}(\rho)=0$ if and only if $\rho$ is real.
    Second, it is well known that the unique maximally imaginarity state is $|+ \rm i\rangle=\frac{|0\rangle + {\rm i} |1\rangle}{\sqrt{2}}$
    for the imaginarity resource theory \cite{A.H.2018}. Hence, it follows that $M_{\mathrm{Re}}(\rho) \leq M_{\mathrm{Re}}(|+ \rm i\rangle \langle + \rm i|)=1-\frac{\sqrt{2}}{2}$.

    (2) Using $\mathrm{Re}(|\psi \rangle \langle \psi|)=\frac{1}{2}({|\psi\rangle\langle\psi|+|\psi^*\rangle\langle\psi^*|})$ and the expression for fidelity $F(|\varphi \rangle \langle \varphi|, \rho)=\sqrt{\langle\varphi|\rho|\varphi\rangle}$ with any pure state $|\varphi\rangle$ and any density operator $\rho$ \cite{Nielsen}, one obtains $M_{\mathrm{Re}}(|\psi\rangle)=1- \frac{1}{\sqrt{2}} \sqrt{1+ |\langle \psi|\psi^*\rangle|^2}$.

    (3) Recall that in \cite{WuKD20212} the geometric imaginarity of a pure state $|\psi\rangle$ is given as $M_{\mathrm{g}}(|\psi\rangle)=\frac{1}{2}({1- |\langle\psi^*|\psi\rangle|})$.
    Combining with Eq. \eqref{eq Mre pure}, it immediately follows that $ M_{\mathrm{Re}}(|\psi\rangle)=1-\frac{1}{\sqrt{2}}\sqrt{{1+(1-2M_{\mathrm{g}}(|\psi\rangle))^2}}. $

   (4) For the set $\{\Pi_i\}_{i=0}^{d-1}$ of rank-$1$ projection operators with real entries,
 employing the relation
   \begin{eqnarray*}
      \mathrm{Re}\left(\sum\limits_{i=0}^{d-1} \Pi_i\rho \Pi_i\right)=\sum\limits_{i=0}^{d-1} \Pi_i\mathrm{Re}(\rho) \Pi_i
    \end{eqnarray*}
    and
     \begin{eqnarray*}
        F\left( \sum_i \Pi_i \rho \Pi_i,  \sum_i \Pi_i \sigma \Pi_i\right) = \sum\limits_{i} \sqrt{p_i q_i} F\left(\rho_{i}, \sigma _{i}\right),
    \end{eqnarray*}
    with $\rho_{i}=\frac{1}{p_i}\Pi_i \rho \Pi_i$, $p_i=\mathrm{Tr}(\Pi_i \rho \Pi_i)$,  $\sigma_{i}=\frac{1}{q_i}\Pi_i \sigma \Pi_i$, $q_i=\mathrm{Tr}(\Pi_i \sigma \Pi_i)$ \cite{Nielsen} \cite{C.H.Xiong2019},
   we have
     \begin{eqnarray*}
        F\left( \sum_i \Pi_i \rho \Pi_i,  {\rm Re }\left(\sum_i \Pi_i \rho \Pi_i\right)\right) = \sum\limits_{i} p_i  F\left(\rho_{i}, {\rm Re }(\rho _{i})\right).
    \end{eqnarray*}
    Subsequently, we derive that $ M_{{\rm Re}}\left( \sum\limits_i \Pi_i \rho \Pi_i\right) = \sum\limits_{i} p_i  M_{{\rm Re}}  \left(\rho_{i}\right).$

\appendixsection{The proof of Theorem \ref{th:4}}
 For any qubit state  $  \rho=\frac{1}{2} ({I}+\boldsymbol{r}\cdot\boldsymbol{\sigma})$
   with $\boldsymbol{r}=(r_x, r_y, r_z)\in \mathbb{R}^3$, $|\boldsymbol{r}|=\sqrt{r_x^2+ r_y^2+ r_z^2}\leq 1$, its real part state is $ \mathrm{Re}(\rho)=\frac{1}{2}({I}+r_x\sigma_x+r_z\sigma_z).    $
    Furthermore, it can be obtained  that $\mathrm{det}(\rho)=\frac{1-|\boldsymbol{r}|^2}{4}$, $\mathrm{det}(\mathrm{Re}(\rho))=\frac{1-r_x^2-r_z^2}{4}$ and $\mathrm{Tr}(\rho\mathrm{Re}(\rho))=\frac{1+r_x^2+r_z^2}{2}$. Since the fidelity between any qubit states $\rho$ and $\sigma$ is $ F(\rho,\sigma)^2=2\sqrt{\mathrm{det}(\rho)\mathrm{det}(\sigma)}+\mathrm{Tr}(\rho\sigma)$ \cite{L.Zhang2018},
    so we get
    \begin{eqnarray*}
       F(\rho,\mathrm{Re}(\rho))=\frac{1}{\sqrt{2}}\Bigg[{\sqrt{(1-|\boldsymbol{r}|^2)(1-|\boldsymbol{r}|^2+r_y^2)}+1+|\boldsymbol{r}|^2-r_y^2} \Bigg]^\frac{1}{2},
    \end{eqnarray*}
   which gives rise to  $M_{\mathrm{Re}}(\rho)=1-\frac{1}{\sqrt{2}}\Bigg[{\sqrt{(1-|\boldsymbol{r}|^2)(1-|\boldsymbol{r}|^2+r_y^2)}+1+|\boldsymbol{r}|^2-r_y^2}\Bigg]^\frac{1}{2}$.

\appendixsection{The proof of Theorem \ref{theo:geo measure}}

  First, recalling the definition of fidelity, we know $  F(\rho,\rho^*)=\mathrm{Tr}\sqrt{\rho^{\frac{1}{2}}\rho^*\rho^{\frac{1}{2}}}$
    and $ F(\rho,\mathrm{Re}(\rho))=\frac{1}{\sqrt{2}}\mathrm{Tr}\sqrt{\rho^2+\rho^{\frac{1}{2}}\rho^*\rho^{\frac{1}{2}}}.$
   For $0 \leq p \leq 1$, the function $f(x)=x^p$ is operator concave \cite{E.Carlen2010}, i.e. for any quantum states $\rho_1, \rho_2 $, and $0 < t<1$, it has $ f\big((1-t)\rho_1+t\rho_2\big) \geq (1-t)f(\rho_1)+tf(\rho_2).$
 Let $t=\frac{1}{2}$ and $ p=\frac{1}{2}$, then we have $ \sqrt{\frac{\rho_1+\rho_2}{2}} \geq \frac{1}{2}\sqrt{\rho_1}+\frac{1}{2}\sqrt{\rho_2},$
  and the equality holds if and only if $\rho_1=\rho_2$.
   Therefore,
   \begin{equation}  \label{eq Tr}
       \mathrm{Tr}\sqrt{\frac{\rho_1+\rho_2}{2}} \geq \frac{1}{2}\mathrm{Tr}\sqrt{\rho_1}+\frac{1}{2}\mathrm{Tr}\sqrt{\rho_2}.
   \end{equation}
    By applying the inequality \eqref{eq Tr} and $\mathrm{Tr}(\rho)=1$, it follows that
    \begin{eqnarray*}
        F\big(\rho,\mathrm{Re}(\rho)\big) \geq \frac{1}{2}\big(1+F(\rho,\rho^*)\big).
    \end{eqnarray*}
    Therefore, by the analytical formula of the geometric imaginarity $ M_{\mathrm{g}}$ in Eq. (\ref{eq g im def}) and the definition of the imaginarity measure $   M_{\mathrm{Re}}$ in Eq. (\ref{eq def mre}), it is readily to get $ M_{\mathrm{Re}}(\rho) \leq M_{\mathrm{g}}(\rho)$
    and the equality holds if and only if $\rho=\rho^*$ i.e., $\rho$ is a real state.

    Second,  the relation $M_{\mathrm{g}}(\rho) \leq 1-[1- M_{\mathrm{Re}}(\rho)]^2$ is derived directly by $ F(\rho,\mathrm{Re}(\rho)) \leq \max\limits_{\sigma\in {\cal R}} F(\rho,\sigma)$.

\appendixsection{The proof of item (ii) in Theorem \ref{th mg real}}
For any qubit state
    \begin{eqnarray}\label{eq rho bloch}
       \rho=\frac{1}{2} (I+\boldsymbol{r}\cdot\boldsymbol{\sigma})=
     \frac{1}{2}
       \left( \begin{array}{cc}
     1+ r_z& r_x-{\rm i} r_y\\
    r_x+{\rm i} r_y & 1-r_z
    \end{array}
   \right),
  \end{eqnarray}
   with $\boldsymbol{r}=(r_x, r_y, r_z)\in \mathbb{R}^3$, $|\boldsymbol{r}|=\sqrt{r_x^2+ r_y^2+ r_z^2}\leq 1$,
its determinant is $\mathrm{det}(\rho) = \frac{1-\left| \boldsymbol{r} \right|^2}{4}$.
For any real state
\begin{equation}  \label{eq real state sigma}
   \sigma=\frac{1}{2} (I+x \sigma_x+z \sigma_z),
\end{equation}
its determinant is $\mathrm{det}(\sigma) = \frac{1-x^2-z^2}{4}$. By calculation, we get $\mathrm{Tr} (\rho \sigma) = \frac{1 + x  r_x + z r_z }{2}$. According to the relation $ F(\rho,\sigma)^2=2\sqrt{\mathrm{det}(\rho)\mathrm{det}(\sigma)}+\mathrm{Tr}(\rho\sigma)$ \cite{L.Zhang2018}, the fidelity between the qubit state $\rho$ in Eq. \eqref{eq rho bloch} and the real state $\sigma$ in Eq. \eqref{eq real state sigma} is
\begin{eqnarray*}
    F(\rho,\sigma)^2=\frac{1}{2}\sqrt{1-\left| \boldsymbol{r} \right|^2} \sqrt{1-x^2-z^2} + \frac{1}{2}(1+x r_x + z  r_z),
\end{eqnarray*}
where $x^2+z^2 \leq 1$. For the sake of brevity, we denote $ G=F(\rho,\sigma)^2$. Next we aim to find the maximum of the function $ G(x,z)=\frac{1}{2}\sqrt{1-\left| \boldsymbol{r} \right|^2} \sqrt{1-x^2-z^2} + \frac{1}{2}(1+x r_x + z  r_z) $
in the domain $\{ (x,z)|x^2+z^2 \leq 1 \}$.
We separate the discussion into two parts, one is on the boundary of the domain $\{ (x,z)|x^2+z^2 = 1 \}$, and the other is inside the domain $\{ (x,z)|x^2+z^2 < 1 \}$. On the boundary of the domain, by the Lagrange Multiplier Method, we find the maximum of $G$ is
$ \frac{1}{2}(1+\sqrt{r_x^2+r_z^2})$.
Inside the domain, by solving the stationary point of function $G$, we find
the maximum of $G$ is
$\frac{1}{2}(\sqrt{1-r_y^2}+1)$.
By comparing these two cases, we get the maximum of funtion $G$ is $\frac{1}{2}(\sqrt{1-r_y^2}+1)$ which is reached at $ (x_0,z_0)=(\frac{r_x}{\sqrt{1-r_y^2}}, \frac{r_z}{\sqrt{1-r_y^2}}). $
This shows $ \max\limits_{\sigma\in {\cal R}} F^2(\rho,\sigma)=\frac{1}{2}(\sqrt{1-r_y^2}+1)$
and the optimal real state is
$\sigma_\star = \frac{1}{2}(I + x_0  \sigma_x + z_0  \sigma_z)$.

\appendixsection{The proof of Theorem \ref{thm:trace norm}}
\hypertarget{app:F}{} 

Before the proof of Theorem \ref{thm:trace norm}, we need two lemmas below.

\begin{lem} \label{thm:5}
    For any quantum state $\rho$, the imaginarity measure of $M_{\mathrm{Re}}$ is bounded from below by
    \begin{eqnarray*}
   M_{\mathrm{Re}}(\rho)\geq   1-\left[P(\rho)-\frac{1}{4}||\rho-\rho^*||^2_{\mathrm{H-S}}+\sqrt{\left(1-P(\rho)\right)\left(1-P(\rho)+\frac{1}{4}||\rho-\rho^*||^2_{\mathrm{H-S}}\right)} \right]^{\frac{1}{2}},
    \end{eqnarray*}
    and from above by
     \begin{eqnarray*}
      M_{\mathrm{Re}}(\rho)\leq 1-\left[P(\rho)-\frac{1}{4}||\rho-\rho^*||_{\rm{H-S}}^2+2\sqrt{E_2(\rho\mathrm{Re}(\rho))}\right]^{\frac{1}{2}},
   \end{eqnarray*}
    where the Hilbert-Schimidt norm is given by $||A||_{\mathrm{H-S}}=\sqrt{\mathrm{Tr}(A^{\dagger}A )}$, the purity by $P(\rho)=\mathrm{Tr}(\rho^2)$, and $E_2$ denotes the sum of the second-order principal minors of matrix.
    \end{lem}

\begin{proof}
    In \cite{J. A. Miszczak},
one can obtain that the fidelity $F(\rho, \mathrm{Re}(\rho))$ is bounded by
 \begin{equation}   \label{eq the bound of F}
      \left[L\left(\rho,\mathrm{Re}(\rho)\right)\right]^{\frac{1}{2}} \leq F(\rho,\mathrm{Re}(\rho)) \leq   [U(\rho,\mathrm{Re}(\rho))]^{\frac{1}{2}}
    \end{equation}
with
\begin{equation}  \label{eq U1}
     U(\rho,\sigma)= \mathrm{Tr}(\rho\sigma)+\sqrt{(1-\mathrm{Tr}(\rho^2))\left(1-\mathrm{Tr}\left(\sigma^2\right)\right)}
    \end{equation}
and
 \begin{eqnarray*}
        L(\rho,\sigma)= \mathrm{Tr}(\rho\sigma)+\sqrt{2}\sqrt{\left(\mathrm{Tr}(\rho\sigma)\right)^2-\mathrm{Tr}\left(\rho\sigma\rho\sigma\right)}.
    \end{eqnarray*}
The trace $\mathrm{Tr}(\rho\mathrm{Re}(\rho))$ can be rewritten as $ \mathrm{Tr}(\rho\mathrm{Re}(\rho))=\frac{1}{2} \left({\rm Tr} (\rho^2)+ {\rm Tr} (\rho \rho^*)\right)={\rm Tr} (\rho^2)+  {\rm Tr} \left(({\rm Im} (\rho))^2\right)=P(\rho)-\frac{1}{4}||\rho-\rho^*||^2_{\mathrm{H-S}}.$
 Inserting this relation into $ U(\rho,\mathrm{Re}(\rho))$ and $ L(\rho,\mathrm{Re}(\rho))$,
we get
     \begin{eqnarray}\label{eq u rho re}
     U(\rho,\mathrm{Re}(\rho))=P(\rho)-\frac{1}{4}||\rho-\rho^*||^2_{\mathrm{H-S}}+\sqrt{\left(1-P(\rho)\right)\left[1-P(\rho)+\frac{1}{4}||\rho-\rho^*||^2_{\mathrm{H-S}}\right]},
    \end{eqnarray}
and
 \begin{eqnarray}   \label{eq:upper}
        L(\rho,\mathrm{Re}(\rho))= P(\rho)-\frac{1}{4}||\rho-\rho^*||_{\rm{H-S}}^2+2\sqrt{E_2(\rho\mathrm{Re}(\rho))},
    \end{eqnarray}
where we have used the relation $(\mathrm{Tr}A)^2=\mathrm{Tr}(A^2)+2E_2(A)$ for $ L(\rho,\mathrm{Re}(\rho))$,
$E_2(A)$ denotes the sum of the second-order principal minors of matrix $A$ \cite{R.A.Horn}.
    Thus the lower and upper bounds of $ M_{\mathrm{Re}}$ can be derived.
\end{proof}

\begin{lem}  \label{le upper bound of Mg}
    For any density operator $\rho$ in a $d$-dimensional Hilbert space, we have
    \begin{equation} \label{eq maxU}
                  \max\limits_{\sigma \in \mathcal{R}}U(\rho, \sigma)\leq
   \frac{1}{d}+\frac{1}{2d}\sqrt{4(d-1)^2-2d(d-1)\left(\mathrm{Tr}(\rho^2)-\mathrm{Tr}(\rho \rho^*)\right)}.
    \end{equation}
\end{lem}

\begin{proof}
For any quantum state $\rho = \sum\limits_{i,j=0}^{d-1} \rho_{ij} |i\rangle\langle j|$  and any real state $\sigma = \sum\limits_{i,j=0}^{d-1} x_{ij} |i\rangle\langle j|$ with $ x_{ij} \in \mathbb{R}$, we have
\begin{equation}  \label{eq U}
    \begin{split}
      U(\rho,\sigma)&=\sum\limits_{i,j=0}^{d-1}\rho_{ij}x_{ji}+\sqrt{1-\mathrm{Tr}\rho^2}\sqrt{1-\sum\limits_{i,j=0}^{d-1}x_{ij}^2}\\
      &:=f(x_{00},x_{01}, ..., x_{0,d-1},x_{10},x_{11},...,x_{1,d-1}, ...,x_{d-1,d-1}).
    \end{split}
\end{equation}
We estimate the maximum of $f$ with the constraints $ \sum\limits_{i=0}^{d-1} x_{ii}=1$ and $x_{ij}=x_{ji}$ for $\forall i,j$.
By the Lagrange multiplier method, we find the maximum is $\frac{1}{d}+\frac{1}{2d}\sqrt{4(d-1)^2-2d(d-1)\left(\mathrm{Tr}(\rho^2)-\mathrm{Tr}(\rho \rho^*)\right)}$. Therefore $
    \max\limits_{\sigma \in \mathcal{R}}U(\rho, \sigma)\leq
    \frac{1}{d}+\frac{1}{2d}\sqrt{4(d-1)^2-2d(d-1)\left(\mathrm{Tr}(\rho^2)-\mathrm{Tr}(\rho \rho^*)\right)}$.
\end{proof}

Now we are ready to prove Theorem \ref{thm:trace norm}.

   First, we can get $  F\left(\rho,\mathrm{Re}(\rho)\right) \geq \left[P(\rho)-\frac{1}{4}||\rho-\rho^*||_{\mathrm{H-S}}^2\right]^{\frac{1}{2}}$
   by Eq. \eqref{eq:upper}.
Given the relation between the trace norm and Hilbert-Schmidt norm $||A||_{\mathrm{H-S}}^2 \leq ||A||_{\mathrm{tr}}^2$,
   we know $ F\left(\rho,\mathrm{Re}(\rho)\right) \geq \left[P(\rho)-M_{\mathrm{tr}}(\rho)^2\right]^{\frac{1}{2}}.$
This yields to $M_{\mathrm{Re}}(\rho) \leq 1-\left(P(\rho)-M_{\mathrm{tr}}(\rho)^2\right)^{\frac{1}{2}}.$

    Second, we know
    \begin{equation} \label{eq max F upper}
        F(\rho,\mathrm{Re}(\rho)) \leq \left[\frac{1}{d}+\frac{1}{2d}\sqrt{4(d-1)^2-2d(d-1)\left(\mathrm{Tr}(\rho^2)-\mathrm{Tr}(\rho \rho^*)\right)}\right]^{\frac{1}{2}}
    \end{equation}
by Lemma \ref{le upper bound of Mg}. Inserting the relation ${\rm Tr} (\rho^2)-{\rm Tr} (\rho \rho^*)=\frac{1}{2}||\rho-\rho^*||^2_{\mathrm{H-S}}$ into Eq. \eqref{eq max F upper} and employing $\frac{1}{d}||A||_{\mathrm{tr}}^2 \leq ||A||_{\mathrm{H-S}}^2$ for a \( d \times d \) matrix \( A \), we can obtain  $ M_{\mathrm{Re}}(\rho) \geq 1-\left[\frac{1}{d}+\frac{1}{d}\sqrt{(d-1)^2-(d-1)M_{\mathrm{tr}}(\rho)^2}\right]^{\frac{1}{2}}. $

\appendixsection{The proof of Theorem \ref{th comple of qubit}}

    For any qubit state $\rho$, its Bloch representations under three MUBs $B_1$, $B_2$ and $B_3$, are
\begin{equation*}
\begin{array}{rcl}
\frac{1}{2}\left( \begin{array}{cc}
         1+ r_z& r_x-{\rm i} r_y\\
    r_x+{\rm i} r_y & 1-r_z
    \end{array}
   \right),\ \
   \frac{1}{2}\left( \begin{array}{cc}
     1+ r_y& r_z-{\rm i} r_x\\
r_z+{\rm i} r_x & 1-r_y
    \end{array}
   \right),\ \
      \frac{1}{2}\left( \begin{array}{cc}
     1+ r_x& r_y-{\rm i} r_z\\
r_y+{\rm i} r_z & 1-r_x
    \end{array}
   \right),
   \end{array}
\end{equation*}
respectively,
with $|\boldsymbol{r}|^2 \leq 1$. By Theorem \ref{th:4}, the imaginarity of qubit state $\rho$   under the orthonormal bases $B_1$, $B_2$ and $B_3$ are
\begin{eqnarray}  \label{eq MReX1}
    M_{\mathrm{Re}}^{B_1}(\rho)=1-\frac{1}{\sqrt{2}}\Bigg[{\sqrt{(1-|\boldsymbol{r}|^2)(1-|\boldsymbol{r}|^2+r_y^2)}+1+|\boldsymbol{r}|^2-r_y^2}\Bigg]^\frac{1}{2},
\end{eqnarray}
\begin{eqnarray}   \label{eq MReX2}
    M_{\mathrm{Re}}^{B_2}(\rho)=1-\frac{1}{\sqrt{2}}\Bigg[{\sqrt{(1-|\boldsymbol{r}|^2)(1-|\boldsymbol{r}|^2+r_x^2)}+1+|\boldsymbol{r}|^2-r_x^2}\Bigg]^\frac{1}{2},
\end{eqnarray}
and
\begin{eqnarray}   \label{eq MReX3}
    M_{\mathrm{Re}}^{B_3}(\rho)=1-\frac{1}{\sqrt{2}}\Bigg[{\sqrt{(1-|\boldsymbol{r}|^2)(1-|\boldsymbol{r}|^2+r_z^2)}+1+|\boldsymbol{r}|^2-r_z^2}\Bigg]^\frac{1}{2},
\end{eqnarray}
respectively. Therefore, it can be obtained that
\begin{eqnarray*}
    \left(1-M_{\mathrm{Re}}^{B_1}\right)^2+\left(1-M_{\mathrm{Re}}^{B_2}\right)^2+\left(1-M_{\mathrm{Re}}^{B_3}\right)^2 \geq \frac{\sqrt{(1-|\boldsymbol{r}|^2)(3-2|\boldsymbol{r}|^2)}+3+2|\boldsymbol{r|^2}}{2},
\end{eqnarray*}
by utilizing inequality $\sqrt{x+y} \leq \sqrt{x}+\sqrt{y}$ for $x,y \in \mathbb{R^+} \cup \{0\}$. The inequality becomes equality if and only if  $|\boldsymbol{r}|=1$, that is, $\rho$ is pure.

\appendixsection{The proof of Theorem \ref{th qutrit pure trade off mutually bases}}

 First, we show the relation in Eq. (\ref{eq compl qutrit pure state under complete mutually}) for pure state case. For any pure state
    $|\varphi \rangle = r_0e^{\mathrm{i}\alpha_0}|0\rangle+ r_1e^{\mathrm{i}\alpha_1} |1\rangle + r_2e^{\mathrm{i}\alpha_2} |2\rangle $ with $r_i\geq 0$, $r_0^2+r_1^2+r_2^2=1$,
     the imaginaries of $|\varphi \rangle$ under MUBs $Z_1$, $Z_2$, $Z_3$ and $Z_4$ are
    \begin{eqnarray*}
    \begin{array}{rcl}
       M_{\mathrm{Re}}^{Z_1}&=& 1-\frac{1}{\sqrt{2}}\sqrt{{1+\left|r_0^2 e^{2\mathrm{i}\alpha_0}+r_1^2 e^{2\mathrm{i}\alpha_1}+r_2^2 e^{2\mathrm{i}\alpha_2})\right|^2}},\\
       M_{\mathrm{Re}}^{Z_2}&=& 1-\frac{1}{\sqrt{2}}\sqrt{{1+\left|r_0^2 e^{2\mathrm{i} \alpha_0} + 2r_1r_2e^{\mathrm{i} (\alpha_1 + \alpha_2)} \right|^2}},\\
       M_{\mathrm{Re}}^{Z_3}&=& 1-\frac{1}{\sqrt{2}}\sqrt{{1+\left|r_0^2e^{2\mathrm{i} \alpha_0} + 2 \alpha_3 r_1r_2 e^{\mathrm{i} (\alpha_1 + \alpha_2)} \right|^2}},\\
       M_{\mathrm{Re}}^{Z_4}&=& 1-\frac{1}{\sqrt{2}}\sqrt{{1+\left|r_0^2e^{2\mathrm{i} \alpha_0} + 2 \alpha_3^2r_1r_2 e^{\mathrm{i} (\alpha_1 + \alpha_2)} \right|^2}},
    \end{array}
    \end{eqnarray*}
    according to Eq. \eqref{eq Mre pure}. This gives rise to
    \begin{eqnarray*}
    \begin{array}{rcl}
        &&\left(1-M_{\mathrm{Re}}^{Z_1}\right)^2+\left(1-M_{\mathrm{Re}}^{Z_2}\right)^2+\left(1-M_{\mathrm{Re}}^{Z_3}\right)^2+\left(1-M_{\mathrm{Re}}^{Z_4}\right)^2 \\
        &=&
        r_0^2 r_1^2 \cos2(\alpha_0-\alpha_1)+r_0^2 r_2^2 \cos2(\alpha_0-\alpha_2)+r_1^2 r_2^2\cos2(\alpha_1-\alpha_2)\\
        && +2+2r_0^4+\frac{1}{2}\left(r_1^4+r_2^4\right)+6r_1^2 r_2^2 \\
        &>& 2+2r_0^4+\frac{1}{2}\left(r_1^4+r_2^4\right)+5r_1^2 r_2^2-r_0^2 r_1^2 -r_0^2 r_2^2.
    \end{array}
    \end{eqnarray*}
In light of the minimum of the function $f(r_0, r_1, r_2)=2+2r_0^4+\frac{1}{2}\left(r_1^4+r_2^4\right)+5r_1^2 r_2^2-r_0^2 r_1^2 -r_0^2 r_2^2$ under the constraint $r_0^2+r_1^2+r_2^2=1$ is $\frac{31}{14}$ by the Lagrange Multiplier Method.
Therefore we have $ \sum\limits_{k=1}^4 \left(1-M_{\mathrm{Re}}^{Z_k}\right)^2 > \frac{31}{14}$ for any pure state.

Second, for any mixed state $\rho$, consider the spectral decomposition $\rho=\sum \limits_{j=1}^3 \lambda_j |\psi_j\rangle\langle\psi_j|$ with eigenvalues $\{\lambda_j\}$ and eigenbasis $\{|\psi_j\rangle\}$. Then the imaginaries of $\rho$ under  MUBs $Z_1$, $Z_2$, $Z_3$ and $Z_4$ satisfies
    {\renewcommand{\arraystretch}{1.6} 
    \begin{eqnarray*}
    \begin{array}{rcl}
       && \sum\limits_{k=1}^4 \left(1-M_{\mathrm{Re}}^{Z_k}(\rho)\right)^2 \\
       & \geq & \sum\limits_{k=1}^4  \left(1-\sum \limits_{j=1}^3 \lambda_j M_{\mathrm{Re}}^{Z_k}(|\psi_j\rangle\langle\psi_j|)\right)^2 \\
       & \geq &  \sum \limits_{j=1}^3 \lambda_j^2 \sum\limits_{k=1}^4  \left(1-M_{\mathrm{Re}}^{Z_k}(|\psi_j\rangle\langle\psi_j|)\right)^2 \\
       & > & \frac{31}{14} \sum \limits_{j=1}^3 \lambda_j^2 = \frac{31}{14}P(\rho),
    \end{array}
    \end{eqnarray*}}
    where the first and second inequalities are derived from the convexity of $M_{\mathrm{Re}}$ and the square function respectively.

\end{document}